\documentclass[11pt]{article}
\usepackage{fullpage}

\usepackage{amsmath,amssymb,graphicx,amsthm, url}

\newcommand{\ignore}[1]{}

\newtheorem{thm}{Theorem}
\newtheorem{lemma}[thm]{Lemma}
\newtheorem{definition}[thm]{Definition}

\newcommand{\N}{\ensuremath{\mathbb{N}}}
\newcommand{\Z}{\ensuremath{\mathbb{Z}}}
\newcommand{\bool}{\ensuremath{\{0,1\}}}
\newcommand{\ineq}{\ensuremath{A\vec{x}\le\vec{b}}}

\begin{document}

\title{On Lower Bound Methods for Tree-like Cutting Plane Proofs}
\author{
Daniel Apon\thanks{Dept.\ of Computer Science, University of Maryland. Email: dapon@cs.umd.edu.}
}

\maketitle

\begin{abstract}
In the book Boolean Function Complexity by Stasys Jukna~\cite{Jukna12}, two lower bound techniques for Tree-like Cutting Plane proofs (henceforth, ``Tree-CP proofs'') using Karchmer-Widgerson type communication games (henceforth, ``KW games'') are presented: The first, applicable to Tree-CP proofs with \emph{bounded} coefficients, translates $\Omega(t)$ \emph{deterministic} lower bounds on KW games to $2^{\Omega(t/\log n)}$ lower bounds on Tree-CP proof size. The second, applicable to Tree-CP proofs with \emph{unbounded} coefficients, translates $\Omega(t)$ \emph{randomized} lower bounds on KW games to $2^{\Omega(t/\log^2 n)}$ lower bounds on Tree-CP proof size.

The textbook proof in the latter case uses a $O(\log^2 n)$-bit randomized protocol for the \emph{GreaterThan} function. However in~\cite{Nisan93}, Nisan mentioned using the ideas of~\cite{FPRU90} to construct a $O(\log n + \log (1/\epsilon))$-bit randomized protocol for \emph{GreaterThan}. Nisan did not explicitly give the proof, though later results in his paper assume such a protocol.

In this short exposition, we present the full $O(\log n + \log (1/\epsilon))$-bit randomized protocol for the \emph{GreaterThan} function based on the ideas of~\cite{FPRU90} for ``noisy binary search.'' As an application, we show how to translate $\Omega(t)$ \emph{randomized} lower bounds on KW games to $2^{\Omega(t/\log n)}$ lower bounds on Tree-CP proof size in the \emph{unbounded coefficient} case. This equates randomness with coefficient size for the Tree-CP/KW game lower bound method.

We believe that, while the $O(\log n + \log (1/\epsilon))$-bit randomized protocol for \emph{GreaterThan} is a ``known'' result, the explicit connection to Tree-CP proof size lower bounds given here is new.
\end{abstract}

\section{Introduction}

In this short exposition, we present a \emph{lower bound method} for Tree-like Cutting Planes proofs, without restriction on the size of coefficients in the proof. Our intention is to both clarify and strengthen similar lower bound methods presented in Chapter 19 of Stasys Jukna's book, Boolean Function Complexity~\cite{Jukna12}. In particular, Jukna shows two lower bound methods. The first technique, for any fixed predicate $f$, maps $\Omega(t)$ \emph{deterministic} lower bounds on Karchmer-Widgerson type communication games for $f$ to $2^{\Omega(t/\log n))}$ lower bounds on the size of Tree-like Cutting Planes proofs of $f$ with \emph{bounded coefficients}. The second technique is a modification of the first, yielding that $\Omega(t)$ \emph{randomized} lower bounds on Karchmer-Widgerson type communcation games for $f$ map to $2^{\Omega(t/\log^2 n)}$ lower bounds on the size of Tree-like Cutting Planes proofs of $f$ with \emph{arbitrarily large coefficients}.

Using the existence of a $O(\log n + \log (1/\epsilon))$-bit randomized protocol for the \emph{GreaterThan} function (henceforth, $GT_n$ for inputs of bit-length $n$), we show how to improve the latter. Namely, we show how $\Omega(t)$ \emph{randomized} lower bounds on Karchmer-Widgerson type communication games for some $f$ map to $2^{\Omega(t/\log n))}$ lower bounds on the size of Tree-like Cutting Planes proofs of $f$ with \emph{arbitrarily large coefficients}. We define these specific terms below. Our definition of communication complexity is standard; see~\cite{KN96} or~\cite{Jukna12} for a general introduction.

In Section 2, we prove the necessary lemmas to establish the $O(\log n + \log (1/\epsilon))$-bit randomized protocol for $GT_n$. In Section 3, we give our main claim: the full lower bound method and proof of correctness.

\section{Lemmas in Communication Complexity}

We begin with definitions of the communication problems we are interested in and lemmas bounding their communication complexity. Nisan~\cite{Nisan93} originally suggested using the ideas of~\cite{FPRU90} to obtain an efficient $GT_n$ protocol. For completeness, we develop the full $GT_n$ protocol and proof. All randomized communication is in the public-coin model (Both players see all random coin flips).

In what follows, we use the existence of a $O(\log n)$-bit protocol with error $O(1/ n)$ for $GT_n$ to give an improved lower bound for Karchmer-Wigderson type communication games.

\begin{definition}[$EQ_n$]
Alice gets $x\in\bool^n$, Bob gets $y\in\bool^n$, and they want to decide if $x = y$.
\end{definition}

\begin{definition}[$GT_n$]
Alice gets $x\in\bool^n$, Bob gets $y\in\bool^n$, and they want to decide if $x > y$.
\end{definition}

\begin{lemma}\label{lemma:EQ-upper-bound}
$EQ_n$ has a randomized protocol with error $\epsilon$ and communication $O(\log(1/\epsilon))$.
\end{lemma}

\begin{lemma}\label{lemma:GT-upper-bound}
$GT_n$ has a randomized protocol with error $\epsilon$ and communication $O(\log n + \log(1/\epsilon))$.
\end{lemma}

\begin{proof}[Proof of Lemma~\ref{lemma:EQ-upper-bound}]
The (public) random coins specify $k=\log(1/\epsilon)$ random $n$-bit strings $r_1, ..., r_k$. Each party computes inner products of its input with each of the $r_i$, and compare inner products modulo 2. This requires Alice sending $k=\log(1/\epsilon)$ bits to Bob. If $x=y$, the inner products are all equal. If $x\ne y$, then the inner products will all be equal with probability $2^{-k}=\epsilon$.
\end{proof}

\begin{proof}[Proof of Lemma~\ref{lemma:GT-upper-bound}]
First, we describe a randomized, binary search type protocol for $GT_n$ with $\epsilon$ error and $O(\log^2 (n/\epsilon))$ communication. Then, we show how to modify this protocol to achieve a randomized protocol for $GT_n$ with $\epsilon$ error and $O(\log n + \log (1/\epsilon))$ communication. Both protocols use, as a subroutine, the protocol of Lemma~\ref{lemma:EQ-upper-bound} -- a randomized protocol for $EQ_m$ with $\delta$ error and $O(\log (1/\delta))$ communication, for some $\delta = \delta(\epsilon)$ to be determined later.\medskip

\noindent
{\bf The $O(\log^2 (n/\epsilon))$ protocol.} Alice and Bob perform a binary search for the most significant bit in which $x = x_1x_2\cdots x_n$ and $y = y_1y_2\cdots y_n$ differ, beginning at position $\lceil n/2\rceil$. During the search, at the position of bit $j\in [n]$, Alice and Bob run
\[
EQ_{n-j+1}(x_ix_{i+1}\cdots x_j, y_iy_{i+1}\cdots y_j)
\]
for some position $i$ depending on $j$ and the transcript of the search up to this point. E.g. $i = 1$ until the search moves right after some test, etc. The search moves in the natural way -- left if the left-half is unequal, and right otherwise.

In total, this costs $O(\log(1/\delta))\log n$ communication for error $\delta$ independently in each invocation of the $EQ$ protocol. This protocol only succeeds iff \emph{all} of the $\log n$ invocations of $EQ$ succeed. To manage the accumulation of error across the $\log n$ rounds, we set $\delta := \epsilon/2n$ to obtain total error of $(\epsilon/2n)^{\log n} < .8\epsilon < \epsilon$. Thus, the total cost is $O(\log(n/\epsilon))\cdot\lceil\log n\rceil = O(\log^2 (n/\epsilon))$ as desired.\medskip

\noindent
{\bf The $O(\log n + \log (1/\epsilon))$ protocol.} To improve the above protocol, consider the depth $\log n$ protocol tree $T$ that describes every path the above $O(\log^2 (n/\epsilon))$ protocol may take. Attach to each leaf (corresponding to the points where the above protocol halts) a descending chain of $c\log(1/\epsilon)$ nodes, for some constant $c\in\N$.

Now, Alice and Bob perform a random, descending walk on $T$ of length $m$. Upon entering every node $u$, they run the corresponding $EQ$ protocol to verify they have moved the correct direction using -- in contrast to the $O(\log^2 (n/\epsilon))$ protocol -- error rate $\delta := 1/4$, which costs 2 bits per invocation. If they detect an error, they backtrack to the parent of $u$ and continue as if having just entered $u$'s parent.

Moreover, upon arriving at a chain (or specifically, at a given leaf node at the head of a chain), they first transmit the single bit corresponding to the supposed most-significant bit in which $x$ and $y$ differ. Note that this index is the label of the matching leaf node by construction. Denote some such, fixed index as $i$. If $x_i$ and $y_i$ are the same, they backtrack. Otherwise, they proceed in subsequent rounds by moving down the chain one step and (repeatedly) running the protocol for $EQ_{i-1}(x_1x_2\cdots x_{i-1}, y_1y_2\cdots y_{i-1})$. If they ever detect inequality to the left of index $i$, they backtrack correspondingly.

Let $m_f + m_b = m$, where $m_f$ is the number of forward steps and $m_b$ is the number of backward steps. We want to show that by setting $m$ large enough, we have 
\begin{itemize}
\item with probability at least $1-\epsilon$, $m_f - m_b > \log n$, implying we have reached a node in one of the chains; and
\item except with arbitrarily low probability, if we halt on a chain, it is the correct one.
\end{itemize}

By Chernoff's bound, setting $m := c'\log(n/\epsilon) = c'(\log n + \log(1/\epsilon))$ for $c' < c$ and $c, c'$ sufficiently large immediately gives the first property. The second property follows by observing that, (i) if the parties reach the correct chain they never leave it, and (ii) if they reach an incorrect chain, they leave it within the next $c'' < c'$ steps with probability $1-(1/4)^{c''}$, so by adjusting $c'$ sufficiently large, we can reduce the chance of reaching the wrong chain to a arbitrarily low chance. Hence, the constants $c, c', c''$ can be jointly adjusted sufficiently large so that the total error is at most $\epsilon$. 

Finally, Alice and Bob send $2$ bits in each of the $m$ rounds, giving total communication of $2m = O(\log n + \log(1/\epsilon))$ as desired.
\end{proof}

\section{Cutting Planes and Karchmer-Widgerson Games}

The Cutting Plane proof system can be viewed as a ``high-dimensional geometric generalization'' of the more well-known Resolution proof system, modeling the ability of an integer program (rigorously defined) to detect unsatisfiable instances given as input. The idea is to start with an integer program \emph{defined} by an \emph{unsatisfiable} system of inequalities \ineq\ with integer coefficients, and use a few basic deductive rules to prove that the set of ``cutting planes'' in high-dimensional space, \ineq, does not have a 0-1 solution. We refer the reader to the book Boolean Function Complexity by Jukna~\cite{Jukna12} for a thorough treatment of the essentials of Cutting Planes and Resolution. We note that the Cutting Plane proof system is strictly stronger than Resolution. Here we give a brief introduction to Cutting Planes proofs (henceforth, ``CP proofs'').

\begin{definition} A \emph{CP Proof that $\ineq\not\in$ SAT} is a directed, acyclic graph composed of linear inequalities arranged into consecutive rows (or ``lines'') such that,
\begin{enumerate}
\item On the first line, you have a list of each singleton inequality in \ineq, called \emph{axioms};

\medskip\noindent\emph{And on each subsequent line:}
\item \emph{(Inequality Addition.)} If $\vec{a}_1\cdot\vec{x}\le c_1$ and $\vec{a}_2\cdot\vec{x}\le c_2$ are on prior lines, then $(\vec{a}_1+\vec{a}_2)\cdot\vec{x}\le c_1+c_2$ can be on this line;
\item \emph{(Scalar Multiplication.)} If $\vec{a}\cdot\vec{x}\le c$ is on a prior line and $d\in\N$, then $d(\vec{a}\cdot\vec{x})\le dc$ can be on this line (Also, if $d\in\Z/\N$, then reverse the inequality; i.e. $d(\vec{a}\cdot\vec{x})\ge dc$ can be on this line);
\item \emph{(Rounded Division.)} If $c(\vec{a}\cdot\vec{x})\le d$ is on a prior line, then $\vec{a}\cdot\vec{x}\le\lfloor\frac{d}{c}\rfloor$ can be on this line;

\medskip\noindent\emph{Finally:}
\item The last line is a single, arithmetically false statement, e.g. $1\le 0$.
\end{enumerate}
\end{definition}

Note that a CP proof of $\ineq\not\in SAT$ implies, indeed, $\ineq\not\in SAT$; the converse is also true. The following definition is standard.

\begin{definition} A \emph{Tree-like Cutting Planes Proof} (henceforth, ``Tree-CP proofs'') is a CP proof whose underlying graph is a tree.
\end{definition}

The connection between lower bounds on Tree-CP proof size and communication complexity comes by the so-called ``KW games,'' or \emph{Karchmer-Wigderson games}~\cite{KW90}. In fact, many -- if not all -- of known exponential lower bounds on Tree-CP proofs of various statements use KW games as the underlying lower bound method.

\begin{definition} A \emph{KW game for \ineq} is a communication complexity game of the following form: Given an unsatisfiable system \ineq, fix a partition of its variables into two parts, $P_1, P_2$. For an assignment $\alpha\in\bool^n$, Alice gets the projection of $\alpha$ onto $P_1$, Bob gets the projection of $\alpha$ onto $P_2$, and their goal is to find an inequality \emph{falsified} by $\alpha$.
\end{definition}

\subsection{An Improved Lower Bound Method for Tree-CP Proofs}

Now we will show how communication complexity lower bounds in KW games translate to Tree-CP size lower bounds. The intention of this first definition is that the threshold decision tree corresponding to a given KW game for some fixed \ineq\ can be viewed as a ``search tree'' solving the KW game.

\begin{definition} A \emph{threshold decision tree} is a rooted, directed tree whose vertices are labeled by (degree-1) \emph{threshold functions} $f$ defined by the property
\[
f(x) = 1 \text{ if and only if } a_1x_1 + \cdots + a_nx_n \le b
\]
with integer coefficients $a_1, ..., a_n, b$, and edges labeled with either 0 or 1. The leaves of the tree are labeled with axioms.
\end{definition}

The following lemma is found in Jukna's~\cite{Jukna12} introduction to CP proof size lower bound methods.

\begin{lemma}[\cite{Jukna12}]\label{lemma:TCPtoTDT} If an unsatisfiable system \ineq\ has a Tree-CP proof of size $S$, then it has a threshold decision tree of depth $D = O(\log S)$.
\end{lemma}

The following lemma is an adaption to our purposes of Lemma 19.11 in Jukna's book~\cite{Jukna12}. Note that when all coefficients in a Tree-CP proof are $poly(n)$-bounded, there is a simple $O(\log n)$-bit, deterministic protocol: Alice just sends the partial sum of her partition to Bob.

\begin{lemma}\label{lemma:TFcomm}
Let $f : \bool^n\rightarrow\bool$ be a (degree-1) threshold function (with, possibly, unbounded coefficients). Then, under any partition, $f$ can be computed with error $O(1/n)$ by a randomized protocol with communication $O(\log n)$.
\end{lemma}
\begin{proof} We use the fact that \emph{any} degree-1 threshold function of $n$ variables can be computed as a threshold function with integer weights of magnitude at most $2^{O(n\log n)}$~\cite{Mur71}, and this is tight up to the $O(\cdot)$~\cite{Has94}.

Thus, each number can be expressed using $m\le O(n\log n)\le O(n^2)$ bits. Each player $P_i$ computes their local sum $x_i$ (over their projection of $\alpha$), then using Lemma~\ref{lemma:GT-upper-bound}, the players can decide if $x_1 > b - x_2$ with error $\epsilon = O(1/n)$ using communication $O(\log m + \log(1/\epsilon)) = O(\log n)$.
\end{proof}

We are now able to prove the main claim. We derive our phrasing of the theorem from its predecessor, Lemma 19.8 of Jukna's book~\cite{Jukna12}.

\begin{thm}\label{lemma:TCP-lower-bound}
If for some partition of the $n$ variables, the KW game for \ineq\ requires $t$ bits of randomized communication, then any Tree-CP proof of \ineq\ (in particular, even with \emph{unbounded coefficients}) must have size $2^{\Omega(t/\log n)}$.
\end{thm}
\begin{proof}
Suppose that the given, unsatisfiable system of inequalities, \ineq\, has a Tree-CP proof (with, possibly, unbounded coefficients) of size $S$. By Lemma~\ref{lemma:TCPtoTDT}, \ineq\ has a threshold decision tree of depth $D = O(\log S)$ solving the KW game for \ineq. By Lemma~\ref{lemma:TFcomm}, for error $O(1/n)$, the threshold functions at each node of the tree can be evaluated with communication $c = O(\log n)$. Thus, the total communication complexity of evaluating the threshold decision tree with probability $1-o(1)$ is at most $cD$. By assumption, at least $t$ bits of communication are necessary, so we conclude
\begin{align*}
cD &\ge t\\
O(\log n\log S) &\ge t\\
S &\ge 2^{\Omega(t/\log n)}.\qedhere
\end{align*}
\end{proof}

\section{Acknowledgements}

We thank William Gasarch and Stasys Jukna for useful discussions and proofreading.

\def\shortbib{0}

\end{document}